\documentclass[12pt]{article}

\usepackage{graphicx}
\usepackage{subfigure}
\usepackage{epstopdf}
\usepackage{color}

\usepackage[normalem]{ulem}

\usepackage{amsmath,amsthm,amsopn,amssymb,epsfig,psfrag,mathtools}
\allowdisplaybreaks

\newtheorem{Theorem}{Theorem}[section]

\newtheorem{Remark}{Remark}[section]

\numberwithin{equation}{section} \numberwithin{figure}{section}

\def\ba{\begin{array}}
\def\ea{\end{array}}

\usepackage{amsthm}

\begin{document}

\title{The final size and critical times of an SIVR epidemic model\thanks{The research
 was partially supported by the National Natural Science Foundation of China (No. 12271470 and 12101301).}}
\date{\empty}
\author{Phyu Phyu Win$^{1,2}$, Zhigui Lin$^1$\thanks{Corresponding author. Email: zglin@yzu.edu.cn (Z. Lin).} and Mengyun Zhang$^3$ \\
{\small $^1$ School of Mathematical Science, Yangzhou University, Yangzhou 225002, China}\\
{\small $^2$ Department of Mathematics, Banmaw University, Banmaw, Myanmar}\\
{\small $^3$ School of Applied Mathematics, Nanjing University Of Finance $\&$ Economics,}\\
{\small  Nanjing 210003, China}
}

\maketitle
\begin{quote}
\noindent
{\bf Abstract.} { 
\small To understand the impact of vaccination, we consider a SIVR  (susceptible-infected-vaccinated-recovered) model which combines impulsive vaccination into the classical SIR model. The final size is firstly defined and estimated, and then the peak value and peak time are considered. Finally, the critical times for a given infected number is studied, and it can be used to define and estimate the stopping time. Our results extend those for the well-understood SIR model.  }

\end{quote}

\section{Introduction}
\smallskip

Infectious diseases threaten the world in various ways, and have a significant impact on human health, economics, social structures, and more \cite{DD}. Controlling infectious diseases has always been a global concern.

Vaccination is an effective strategy for minimizing of infectious diseases and has been considered in many literatures. The SIR (susceptible-infected-recovered) models and its' variants have been extensively explored in \cite{ADO, SLZ}. Gao et al. \cite{GO} investigated SIRS epidemic model with seasonal varying contact rate and mixed vaccination strategy. They studied the permanence and extinction of the disease for SIRS model.  The standard SIR epidemic model was extended to a fourth compartment, $V$, of vaccinated persons in \cite{SK}.   The influence of vaccinations on the total cumulative number  is calculated by comparing with monitored real time Covid-19 data in different countries. The reduction in the final cumulative fraction of infected persons  is given by using the actual pandemic parameters.
In \cite{MTur}, Turkyilmazoglu revealed the vaccine application offers less control over the spread of virus, since some portion of vaccinated people is not totally protected/immuned and viable to infection again after a while due to weak/loss immunity offered by the vaccine. See also \cite{KV, LBZ} and references therein.

As we know, the final size is one of the most concerning indicators for an epidemic.
The final size of an epidemic is the total proportion of the population that becomes infected \cite{BCA, MSW}. Knowing the final size is crucial for assessing the severity of epidemics, evaluating the impact of interventions, guiding healthcare planning, and informing public health decision-making. To understand the development trend and severity of infectious diseases, researchers have estimated and analyzed the final size by using mathematical models and empirical data \cite{CPS, CUI, FB, CG, LZL, JCM}. Results tell us that public health measures \cite{CPS, JCM}, vaccination campaigns \cite{JLS,MTur}, and other interventions \cite{DAAD, JAA} all can reduce the final size by limiting the spread of the disease.

Besides the final size, the peak amplitude of an epidemic and the peak time \cite{Tur} become the main concerns of CDC staff. Understanding the peak value and peak time helps optimize healthcare responses, minimize the burden on healthcare systems, and reduce the overall morbidity and mortality associated with the infectious disease.

The classical SIR model involves the system of ODE \cite{HIP}
\begin{eqnarray}\label{a0}
\left\{
\begin{array}{lll}
\frac{\textrm{d}S(t)}{\textrm{d}t}=-\beta S(t) I(t),\\
\frac{\textrm{d}I(t)}{\textrm{d}t}=\beta S(t) I(t)-\gamma I(t)
\end{array} \right.
\end{eqnarray}
with the initial conditions
$$S(0)=S_{0},I(0)=I_{0},S_{0}+I_{0}=N,$$
where $S(t)$ and $I(t)$ denote the susceptible and infected compartments of a given population in the presence of an infectious disease. The constant $N$ is the size of the population, $R(t):=N-S(t)-I(t)$ is the recovered compartment of the population at time $t$. The positive parameters $\beta$ and $\gamma$ are the infected and recovery rates per unit time, respectively. It is well-known that the basic reproduction number is $\mathcal{R}_{0}=\frac{\beta}{\gamma}N$.

To introduce vaccination into the SIR model, we first make the following assumptions:

1. Population is partitioned into four classes, the susceptibles $S(t)$, infectious $I(t)$, vaccinated $V(t$) and recovered $R(t)$ respectively, see the flow chart in Fig. \ref{tu4}. The total population size is $N$, that is, $S(0)+I(0):=N$.

2. The factor $\sigma (0\leq\sigma\leq 1)$ is the infection probability of the vaccinated member contacting with the infections, $\sigma=0$ means that the vaccine is completely effective in preventing infection and $\sigma=1$ means that the vaccine has no effect.

3. The vaccinations are implemented at times $t=n\tau,n=1,2,3, \cdots$, $\varphi (0<\varphi<1)$ is pulse vaccination rate, and the interpulse time, i.e., the time between two consecutive pulse vaccinations is $\tau$.

4. The vaccine wears off at a constant rate $\theta$.

\begin{figure}[ht]
\centering
\includegraphics[width=0.4\textwidth]{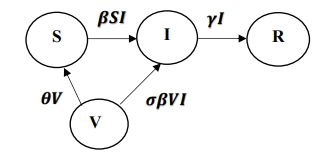}
\caption{Flow chart of SIVR epidemic model.}
\label{tu4}
\end{figure}

When pulse vaccination is incorporated into the SIR model (\refeq {a0}), the system becomes SIVR (susceptible-infected-vaccinated-recovered) epidemic model
\begin{equation}\label{a1}
		\begin{split}
			\begin{cases}\begin{rcases}
		      \frac{\textrm{d}S(t)}{\textrm{d}t}=-\beta S(t) I(t)+\theta V(t),\\
                \frac{\textrm{d}I(t)}{\textrm{d}t}=\beta S(t) I(t)+\sigma\beta V(t) I(t)-\gamma I(t) ,\\
                 \frac{\textrm{d}V(t)}{\textrm{d}t}=-\sigma\beta V(t) I(t)-\theta V(t),\end{rcases}t\neq n\tau,n=1,2,3,\cdots,\\
                 \begin{rcases}
                 S(n\tau^{+})=(1-\varphi)S(n\tau),\\
                 I(n\tau^{+})=I(n\tau),\\
                 V(n\tau^{+})=V(n\tau)+\varphi S(n\tau)\\
                 \end{rcases}t=n\tau,n=1,2,3,\cdots,\\
                 S(0)=S_0>0,\ I(0)=I_0>0,\ V(0)=V_0=0.
                 \end{cases}
		\end{split}
	\end{equation}
We are more interested in the final size, the peak value, the peak time and the critical times for the model \eqref{a1}. The paper is structured as follows. The next section deals with the final size, and its estimation is derived. Section~\ref{sec:Peak value of SIVR model} is devoted to the peak value and peak time,  and four critical times are defined and estimated in Section~\ref{sec:The critical times}. The stopping time is defined in the last section and its estimates for the model \eqref{a0} are presented.

\section{The final size}\label{sec:The final size}
\smallskip

 It is easy to see that the solution $(S, I, V)$ to problem \eqref{a1} exists uniquely and is global. Moreover,
 $S(t)+I(t)+V(t)\leq N$ for $t\geq 0$.

We first claim that
 \begin{equation}
    \lim_{t\to\infty}I(t)=0,
    \label{limits}
 \end{equation}
 which means that the disease vanishes eventually, see Figure \ref{tu3}.

\begin{figure}[ht]
\centering
\includegraphics[width=0.5\textwidth]{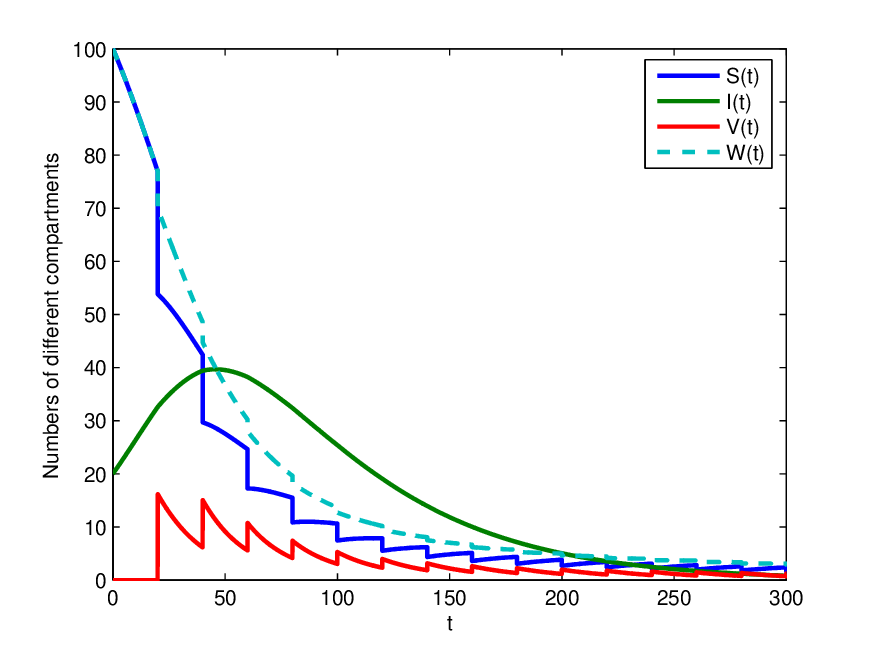}
\caption{\scriptsize The SIVR model. $W(t)=S(t)+V(t)$. The infected number $I$ increases at the beginning, and then decreases in a long run.}
\label{tu3}
\end{figure}

We now prove \eqref{limits} by contradiction. In fact, if
$\limsup_{t\to\infty} I(t)=\varepsilon_{0}$ for some $\varepsilon_{0}>0$, there
exists a sequence $\{t_{n}\}^{+\infty}_{n=1}$ such that $I(t_n)\geq\frac{\varepsilon_{0}}{2}$ for any $n$.

Since $S,I$ are bounded by $N$, we have $|I'(t)|\leq M$ for some $M\geq 0$ for $t\geq 0$ and $t\neq n\tau$.
Therefore, there exists $\delta_{0}>0$ such that
\begin{equation*}
|I(t)-I(t_{n})|\leq\frac{\varepsilon_{0}}{4}\
\text{ for } t\in(t_{n}-\delta_{0}, t_n]\, \textrm{or} \, t\in [t_n, t_{n}+\delta_{0}),
\end{equation*}
which means
\begin{equation*}
I(t)\geq I(t_{n})-|I(t)-I(t_{n})|\geq \frac{\varepsilon_{0}}{2}-\frac{\varepsilon_{0}}{4}=\frac{\varepsilon_{0}}{4}\ \text{ for } t\in(t_{n}-\delta_{0}, t_n]\, \textrm{or}
\, t\in [t_n, t_{n}+\delta_{0})
\end{equation*}
 and
\begin{equation}\label{lll}
\int_{0}^{\infty}I(t)dt\geq\sum_{n=1}^{\infty}\int_{t_{n}-\delta_{0}}^{t_{n}+\delta_{0}}I(t)dt
\geq\sum_{n=1}^{\infty}\frac{\varepsilon_{0}}{4}\delta_{0}=\infty.
  \end{equation}
On the other hand, since
\begin{equation}
    \frac{\textrm{d}}{\textrm{d}t}(S+I+V)(t)=-\gamma I(t),
  \end{equation}
integrating from $t=0$ to $t=+\infty$ yields
\[(S+I+V)(+\infty)-(S+I+V)(0)=-\gamma\int_{0}^{\infty}I(t)dt\]
and
\begin{equation}\label{ee}
     \int_{0}^{\infty}I(t)dt=\frac{-1}{\gamma}[(S+I+V)(\infty)-(S+I+V)(0)]<\infty,
\end{equation}
which leads a contradiction to \eqref{lll}, therefore, $\limsup_{t\to\infty} I(t)={0}$ and $\lim_{t\to\infty} I(t)={0}$.

\bigskip

  Since that $(S+V)'=-\beta S(t) I(t)-\sigma\beta V(t) I(t)\leq 0$ for $t\neq n\tau$, and $(S+V)$ is continuous for $t\geq 0$,  the limit of $(S+V)(t)$ as $t\to \infty$ exists.

As in \cite{MSW} , the final size ($Z$) of the epidemic is defined as the number of members of the population who are infected over the course of the
epidemic. For our model \eqref{a1},
$$Z=N-(S+V)_\infty.$$

Now we estimate $Z$ and define
   $ W(t)=S(t)+V(t)$ for $t\geq 0$.
  It follows from \eqref{a1} that
  \begin{equation} \label{dd}
   -\beta W(t)I(t)\leq W'(t)=-\beta S(t) I(t) -\sigma\beta V(t) I(t) \leq-\sigma\beta W(t)I(t),
            \end{equation}
and dividing \eqref{dd} by $W(t)$ gives
\begin{equation} \label{1dd}
   -\beta I(t)\leq W'(t)/ W(t) \leq-\sigma\beta I(t).
            \end{equation}
Integrating \eqref{1dd} from $0$ to $+\infty$ yields
\begin{equation}\label{ff}
    -\beta\int_{0}^{+\infty}I(t)dt\leq\ln W_{\infty}-\ln W_{0}\leq-\sigma\beta \int_{0}^{+\infty}I(t)dt.
        \end{equation}
By combining \eqref{ee} and \eqref{ff}, we obtain
{\small \[\frac{\beta}{\gamma}[(S+I+V)(\infty)-(S+I+V)(0)]\leq \ln\frac{W_{\infty}}
    {W_{0}}\leq\frac{\sigma\beta}{\gamma} [(S+I+V)(\infty)-(S+I+V)(0)],\]}
which means
    $$\frac{\sigma\beta}{\gamma}[S_{0}+I_{0}-(S+V)(\infty)]\leq\ln\frac{W_{0}}
    {W_{\infty}}\leq\frac{\beta}{\gamma}[S_{0}+I_{0}-(S+V)(\infty)],$$
    and equivalently,
 $$\frac{\sigma\beta}{\gamma}[N-W_{\infty}]\leq\ln\frac{W_{0}}{W_{\infty}}\leq\frac{\beta}{\gamma}[N-W_{\infty}].$$
 In particular, since $\mathcal{R}_{0}=\frac{\beta N}{\gamma}$, we have
    $$\sigma \mathcal{R}_{0}[1-\frac{W_{\infty}}{N}]\leq\ln\frac{W_{0}}{W_{\infty}}\leq \mathcal{R}_{0}[1-\frac{W_{\infty}}{N}].$$
Recalling that $W_0=S_0$, we have estimates for the final size $Z=N-W_\infty$, where $W_\infty$ satisfies
 $$W_*\leq W_\infty\leq W^*,$$
and $W_*$, $W^*$ are, respectively, the unique positive root of
\begin{equation}\label{g1}  g_{1}(x):=\ln\frac{S_{0}}{x}-\mathcal{R}_{0} [1-\frac{x}{N}]=0
 \end{equation}
and
\begin{equation}\label{g2}   g_{2}(x):=\ln\frac{S_{0}}{x}-\sigma \mathcal{R}_{0} [1-\frac{x}{N}]=0.
\end{equation}

\begin{figure}[ht]
\centering
\subfigure[]{ {
\includegraphics[width=0.26\textwidth]{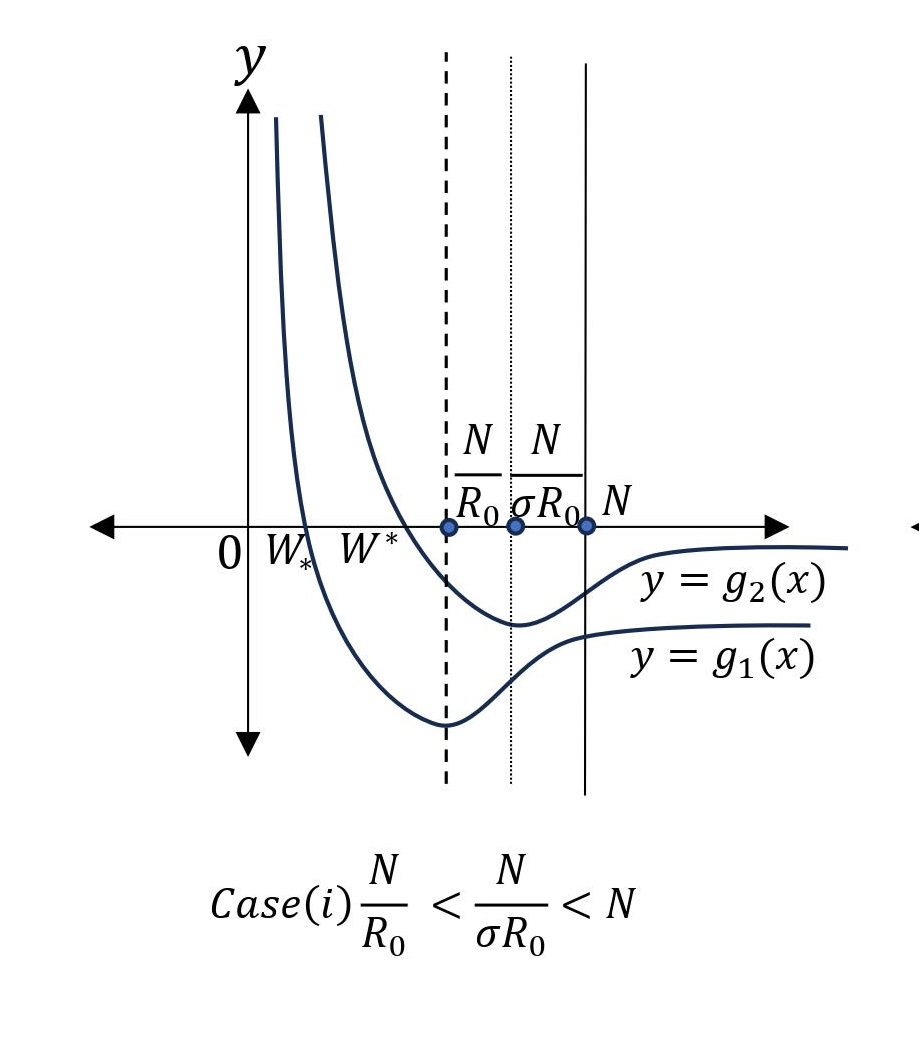}
} }
\subfigure[]{ {
\includegraphics[width=0.26\textwidth]{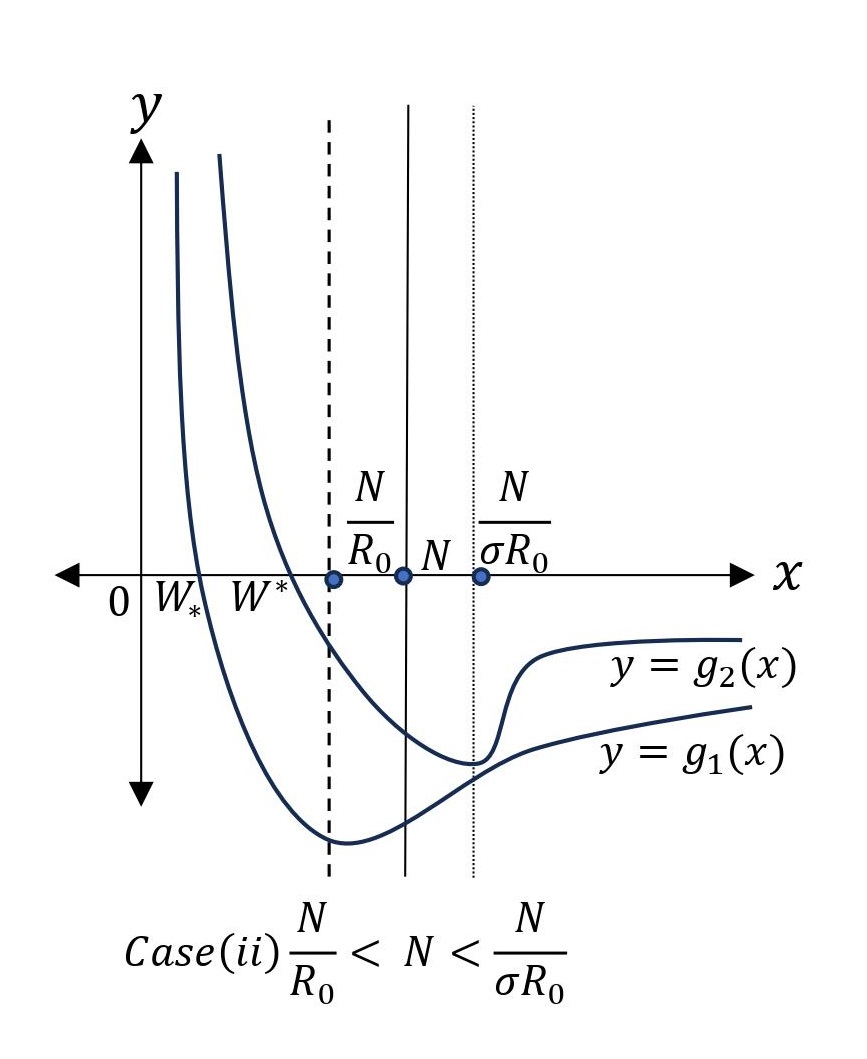}
} }
\subfigure[]{ {
\includegraphics[width=0.26\textwidth]{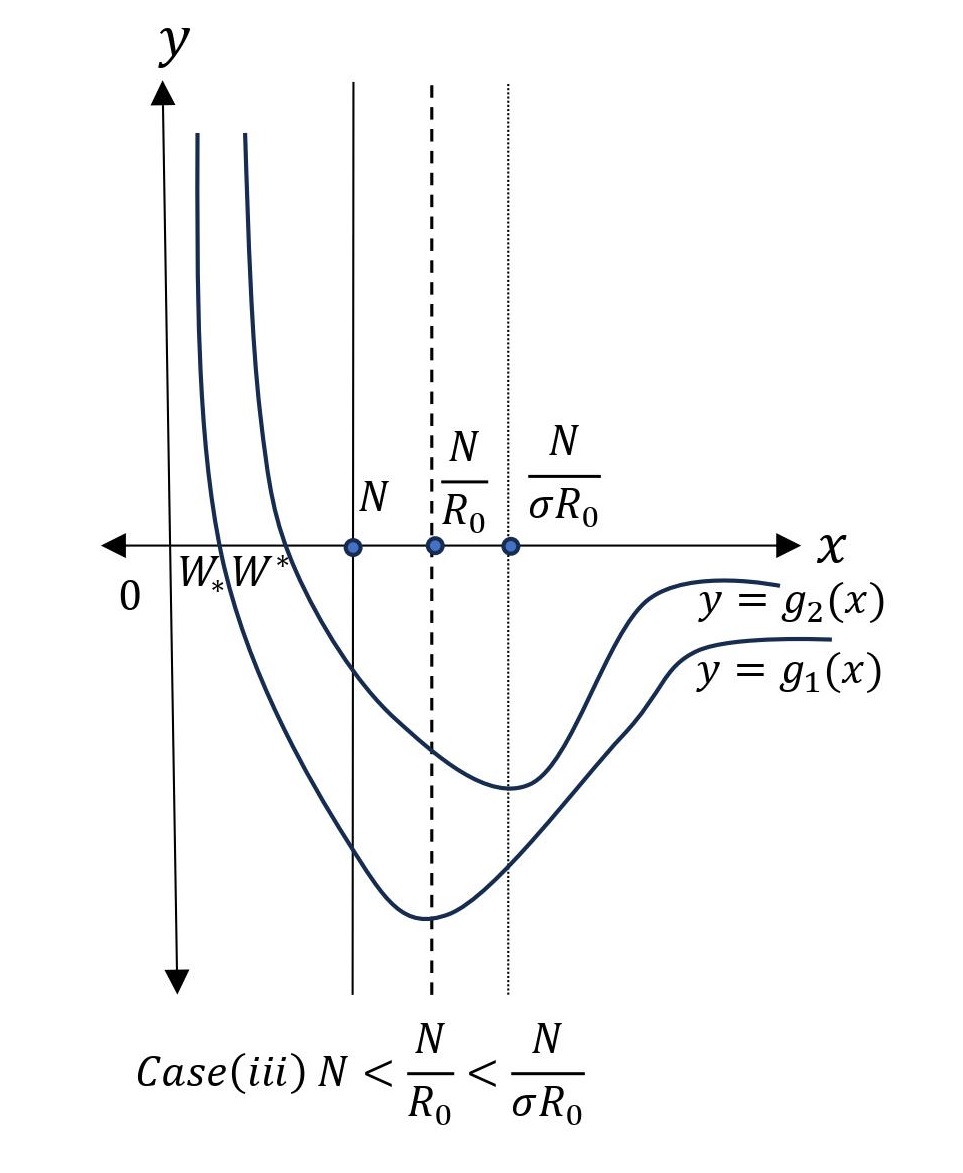}
} }
\caption{\scriptsize The graphs of $g_1$ and $g_2$ for different cases.}\vspace*{-10pt}
\label{tu1}
\end{figure}

In fact, we can see from Fig. \ref{tu1} that
 $$g_{1}(0^{+})>0,\quad g_{1}(x)<0 \text{ for } x\geq N,    $$
there exists a unique $W_{*}$ satisfying $0<W_{*}<N$ and
$$    g_{1}(W_{*})=0,\ g_{1}(x)<0 \text{ for } 0<W_{*}<x,    $$
So we have $$W_{*}\leq W_{\infty} .$$
Similarly,
$$ g_{2}(x)\geq 0 \text{ for } x\leq W_{*},
    $$
there exists a unique $W^{*}$ satisfying $0<W_{*}\leq W^{*}<N$ such that
$$
 g_{2}(W^{*})=0, g_{2}(x)\geq 0 \text{ for } 0<W_{*}\leq x\leq W^{*}\leq N.
    $$
Therefore $$W_{\infty}\leq W^{*}.$$

\begin{Theorem} The final size of the model \eqref{a1} is $Z=N-W_\infty$ and $W_\infty$ satisfies
 $$W_*\leq W_\infty\leq W^*,$$
where $W_*$ is the unique positive root of $g_1=0$ defined in \eqref{g1}, and $ W^*$ is the unique positive root of $g_2=0$ defined in \eqref{g2}.
\label{Theorem 1}
\end{Theorem}

\begin{Remark} If $\sigma=1$, which means that the vaccine has no effect, we then have $g_1=g_2$ and $W_*=W_\infty=W^*,$
and $W_\infty$ is the unique positive root of $g_1=0$, see  Section 9.2  in \cite{BCA}.
\label{Theo1}
\end{Remark}

\section{The peak value and peak time}\label{sec:Peak value of SIVR model}
\smallskip

Noting that $i(0)>0$ and  $\lim_{t\to\infty}I(t)=0$, there exists $t_P\geq 0$ such that $I_m:=I(t_P)=\sup_{t\geq 0} I(t)$,
we usually call $t_P$ {\bf the peak time} and $I(t_P)$ is {\bf the peak value}.

By using equation \eqref{a1} and $W=S+V$, we have
    \[\frac{d}{dt}[(W+I)(t)-\frac \gamma \beta \ln (W(t))]=-\gamma \frac {(1-\sigma)VI}W\leq 0,\]
     \[\frac{d}{dt}[(W+I)(t)-\frac \gamma {\sigma\beta} \ln (W(t))]=\gamma I(t)\frac {(1-\sigma)SI}{\sigma W}\geq 0\]
for $t\neq n\tau,n=1,2,3,...\geq 0$, which implies that
\begin{equation}\label{204}
\frac \beta \gamma  \frac{d}{dt}[(W+I)(t)] \leq \frac{d}{dt}[\ln (W(t))]\leq \frac {\sigma\beta} \gamma \frac{d}{dt}[(W+I)(t)].
\end{equation}
Recalling that $W, I$ are continuous and piecewise differentiable, integrating \eqref{204} from $0$ to $t$ yields
\begin{equation}\label{224}
    \frac{\sigma\beta}{\gamma}[S_{0}+I_{0}-W(t)-I(t]
    \leq \ln\frac{W_{0}}{W(t)}\leq \frac{\beta}{\gamma}[S_{0}+I_{0}-W(t)-I(t)]
\end{equation}
for any $t>0$ and
 \begin{equation}\label{bb1}
    S_{0}+I_{0}-W(t)-\frac{\gamma}{\sigma\beta}\ln \frac{W_{0}}{W(t)} \leq I(t)
    \leq S_{0}+I_{0}-W(t)-\frac{\gamma}{\beta}\ln \frac{W_{0}}{W(t)}.
    \end{equation}

Usually, the maximum number of infectives is the number of infectives when
the derivative of $I$ is zero, but in our model \eqref{a1}, $I(t)$ is not
differentiable at $t=n\tau,n=1,2,3,...$, so we have to consider two cases.

If $t_{P}\neq n\tau$ for $n=1,2,3,...$, we have $I'(t_P)=0$ and
$(S+\sigma V)(t_{P})=\frac{\gamma}{\beta}$. Using \eqref{bb1} with $t=t_P$ yields
    \begin{equation}\label{bb0}
    S_{0}+I_{0}-\frac{\gamma}{\sigma\beta}-\frac{\gamma}{\sigma\beta}\ln W_{0}+\frac{\gamma}{\sigma\beta}
    \ln\frac{\gamma}{\beta}\leq I_{m}
    \leq S_{0}+I_{0}-\frac{\gamma}{\beta}-\frac{\gamma}{\beta}\ln W_{0}+\frac{\gamma}{\beta}\ln\frac{\gamma}{\sigma\beta}
    \end{equation}
since $\frac{\gamma}{\beta}\leq W(t_{p})\leq \frac{\gamma}{\sigma\beta}$.

If $t_{P}=n_0 \tau$ for $n_0=1,2,3,...$. In this case,
 we have $I'(t^{-}_{P})\geq0$ and $I'(t^{+}_{P})\leq0$, that is,
$$(S+\sigma V)(t_{P})\geq \frac{\gamma}{\beta}\ \textrm{and}\ (1-\varphi+\sigma\varphi)S(t_{P})+\sigma V(t_{P})\leq\frac{\gamma}{\beta},$$
and therefore $\frac{\gamma}{\beta}\leq W(t_{P})\leq \frac{\gamma}{\beta}\frac{1}{\min\{1-\varphi+\sigma\varphi,\sigma\}}$, which together with \eqref{bb1} gives
  \begin{eqnarray}\label{bb2}
\begin{array}{lll}
 S_{0}+I_{0}-\frac{\gamma}{\beta}\frac{1}{\min \{(1-\varphi+\sigma\varphi),\sigma\}}-\frac{\gamma}{\beta}\ln W_{0}+\frac{\gamma}{\sigma\beta}\ln\frac{\gamma}{\sigma\beta}
    \leq I_{m}\\
    \leq S_{0}+I_{0}-\frac{\gamma}{\beta}-\frac{\gamma}{\beta}\ln W_{0}+
    \frac{\gamma}{\beta}\ln[\frac{\gamma}{\beta}\frac{1}{\min\{(1-\varphi+\sigma\varphi),\sigma\}}].
\end{array}
\end{eqnarray}
Combining two cases, we have estimates \eqref{bb2} of the peak value since that \eqref{bb2} holds if \eqref{bb0} holds.

We now turn to estimates of the peak time $t_P$. It follows from \eqref{1dd} that
\begin{eqnarray}\label{b1b3}
-\frac{W'}{\beta WI}\leq 1 \leq -\frac{1}{\sigma}\frac{W'}{\beta WI},
\end{eqnarray}
which together with \eqref{bb1} gives
\begin{eqnarray}
\begin{array}{lll}\label{b1b31}
-\frac{W'}{\beta W[S_{0}+I_{0}-W(t)-\frac{\gamma}{\beta}\ln W(0)
    +\frac{\gamma}{\beta}\ln W(t)]}\leq 1 \leq -\frac{1}{\sigma}\frac{W'}{\beta W[S_{0}+I_{0}-
    W(t)-\frac{\gamma}{\sigma \beta}\ln W(0)+\frac{\gamma}{\sigma \beta}\ln W(t)]},
\end{array}
\end{eqnarray}
and by integrating \eqref{b1b31} from $0$ to $t_{P}$, we have
  \begin{eqnarray*}
\begin{array}{lll}  \int_{W_0}^{W(t_{p})}\frac{dz}{\beta z[z-S_{0}-I_{0}+\frac{\gamma}{\beta}\ln W(0)
    -\frac{\mu}{\beta}\ln z]}
    \leq t_{P}
   \leq \int_{W_0}^{W(t_{P})}\frac{dz}{\sigma\beta z[z-S_{0}-I_{0}+\frac{\gamma}{\sigma\beta}\ln W(0)-\frac{\gamma}{\sigma\beta}\ln z]}.
   \end{array}
\end{eqnarray*}
Recalling that
$\frac{\gamma}{\beta}\leq W(t_{P})\leq\frac{\gamma}{\beta}\frac{1}{\min(1-\varphi+\sigma\varphi,\sigma)}$ and $W(0)=S_0$ yields
\begin{eqnarray}\label{bbb}
\begin{array}{lll}
    \int_{S_0}^{\frac{\gamma}{\beta}}\frac{dz}{\beta z[z-S_{0}-I_{0}+\frac{\gamma}{\beta}\ln S_0
    -\frac{\gamma}{\beta}\ln z]}
    \leq t_{p}
   \leq \int_{S_0}^{\frac{\gamma}{\beta}\frac{1}{\min\{1-\varphi+\sigma\varphi,\sigma\}}}\frac{dz}{\sigma\beta z[z-S_{0}-I_{0}+\frac{\gamma}{\sigma\beta}\ln S_0-\frac{\gamma}{\sigma\beta}\ln z]}.
    \end{array}
\end{eqnarray}

\begin{Theorem} The peak value $I_m$ of the model \eqref{a1} satisfies \eqref{bb2}, and the peak time $t_P$ satisfies \eqref{bbb}.
Moreover, if $\sigma=1$, we have
\begin{eqnarray}
\begin{array}{lll}
 I_{m}= S_{0}+I_{0}-\frac{\gamma}{\beta}-\frac{\gamma}{\beta}\ln W_{0}+
    \frac{\gamma}{\beta}\ln \frac{\gamma}{\beta},\\
 t_{p}=\int_{S_0}^{\frac{\gamma}{\beta}}\frac{dz}{\beta z[z-S_{0}-I_{0}+\frac{\gamma}{\beta}\ln S_0-\frac{\gamma}{\beta}\ln z]}.
\end{array}
\end{eqnarray}
\label{Theorem2}
\end{Theorem}

As an example,  we illustrate estimates of the peak time and peak value of SIVR model in Fig. \ref{tu2}. Assume that there is no initially vaccination, the infection probability $\sigma$ of the vaccinated member contacting with the infections is small. Let $\beta=0.0005, \theta=0.03, \sigma=0.002, \gamma=0.01, \varphi=0.3 $, and take $S_{0}=100, I_{0}=20$.  Fig. \ref{tu2} shows that the peak time is $t_{P}=60$ and the peak value is $I_{m}=53$.

\begin{figure}[ht]
\centering
\includegraphics[width=0.5\textwidth]{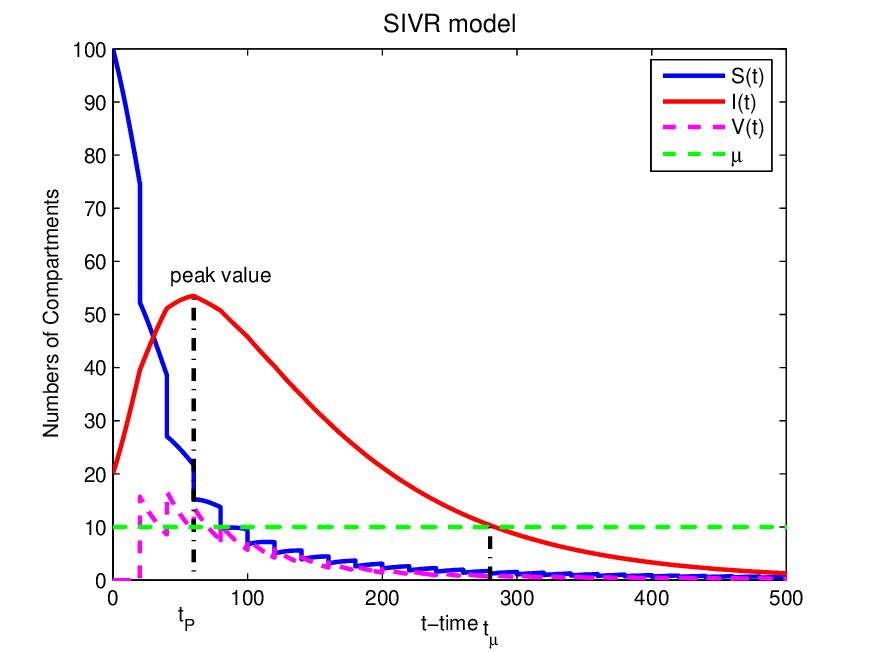}
\caption{\scriptsize The SIVR model and $W(t)=S(t)+V(t)$. The peak value $I_m$ is 53. $t_P(=60)$ is the peak time and $t_\mu(280)$ is the critical time for the infected desired level $\mu=10$. }
\label{tu2}
\end{figure}

\section{The critical times}\label{sec:The critical times}

As in \cite{HIP}, there are two important critical times for the epidemic model. The one is the first time ($t_\mu$) the infected population is below a given threshold ($\mu$), and the other is the first time the number of infected individuals begins to decrease, which is in fact the peak time $t_P$ for most models with single wave and can be viewed as the critical time with $\mu$ replaced by the peak value $I_P$.

However, owing to the complexity of the development of infectious diseases, multiple waves of infections are possible. For example, in several countries successive waves of the COVID-19 disease have been observed \cite{MKJ, RCR}. So we now present the following four critical times.

Let $(S, I, V)$ be the solution of \eqref{a1} with $S(0)=x\geq 0$ and $I(0)=y\geq 0$. For a given threshold $\mu>0$, we denote four critical times as followings, see Fig. \ref{tu5}.

$(i)$ $u_*$ is the first time at which the number of infected individuals is not great than the given value, that is,
$$u_*(x,y):=\underline t_\mu =\inf \{t>0:\ I(t)\leq \mu\};$$

$(ii)$ $u^*$ is the last time at which the number of infected individuals is not less than the given value, that is,
$$u^*(x,y):=\overline t_\mu =\sup \{t>0:\ I(t)\geq \mu\};$$

$(iii)$ $v_*$ is the first time from which the number of infected individuals begins to decreases, that is,
$$v_*(x,y):=\inf \{t>0:\ \exists \delta>0 \, \textrm{such that}\, I(z)< I(t) \,\textrm{for}\, z\in (t, t+\delta)\};$$

$(iv)$ $v^*$ is the last time at which the number of infected individuals is not decreasing, that is,
$$v^*(x,y):=\sup \{t>0:\ \exists \delta>0 \, \textrm{such that}\, I(z)\leq I(t) \,\textrm{for}\, z\in (t-\delta, t)\};$$

\begin{figure}[ht]
\centering
\includegraphics[width=0.5\textwidth]{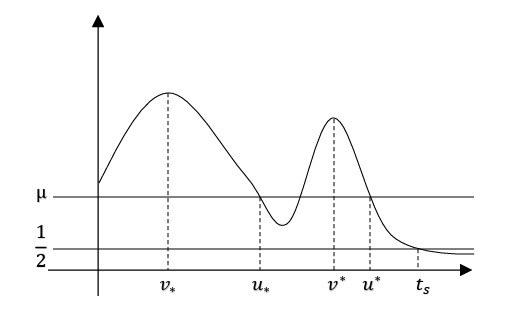}
\caption{\scriptsize The four critical times and the stopping time.}
\label{tu5}
\end{figure}

It is easy to see that $u_*(x,y)=0$ when $0\leq y<\mu$, and $v_*=v^*$ for epidemic models with single wave.
For our model \eqref{a1}, let $\mu=10$, then the first time $\underline t_\mu$ at which the number of infected individuals is not great than $\mu$ is $280$, see Fig. \ref{tu2}, and the last time $\overline t_\mu$ at which the number of infected individuals is not less than $\mu$ is also $280$.

Next we derive some estimates for the critical times.
\begin{Theorem}
For each $x\geq0$ and $y\geq\mu$,
\begin{equation} \label{eke}
\frac{\ln x-\ln \frac{\gamma}{\beta \sigma}}{\beta I_m}\leq u_{*}(x,y)\leq\ \frac{x+y}{\gamma\mu}.\end{equation}
\end{Theorem}

\begin{proof}
Let $S, I$ and $V$ be the solution of the SIVR model \eqref{a1} with $S(0)=x$ and $I(0)=y$.
Since $\frac{\textrm{d}}{\textrm{d}t}(S+I+V)(t)=-\gamma I(t),$
$$\int_{0}^{t}\gamma I(\tau)\textrm{d}\tau+S(t)+I(t)+V(t)=x+y.$$
Choosing $t=u_{*}(x,y)$ and noting that $I(\tau)\geq\mu$ for $0\leq \tau \leq u_{*}(x,y)$ gives
\begin{equation}\label{cc1}
u_{*}(x,y)\gamma\mu\leq\int_{0}^{u_{*}(x,y)} \gamma I(\tau)\textrm{d}\tau\leq x+y,
\end{equation}
and we have the right inequality of  \eqref{eke}.

On the other hand, since that $I(t)>\mu=I(u_{*})$ for $0\leq t<u_{*}$ by the definition of $u_{*}$, we have $I'(u^{-}_{*})\leq 0$, which together with the second equation of \eqref{a1} yields
$$ (S+\sigma V)(u_{*})\leq \frac{\gamma}{\beta}.$$
and therefore,
\begin{equation} \label{ef1} (S+V)(u_{*}):=W(u_*) \leq \frac{\gamma}{\sigma \beta}.\end{equation}
It follows from  \eqref{1dd} that
\begin{equation} \label{ef2}
   W(u_*) \geq  W(0)e^{-\beta\int_{0}^{u_*} I(\tau){\textrm{d}\tau}}\geq  xe^{-\beta u_* I_m},    \end{equation}
which together with \eqref{ef1} give the left inequality of \eqref{eke}.
\end{proof}

\begin{Theorem}
 For each $x\geq \frac{\gamma}{\beta}$ and $y>0$,
\begin{equation}\label{de1}
\frac{\ln x-\ln(\frac{\gamma}{\beta}\frac{1}{\min \{1-\varphi+\sigma\varphi,\sigma\}})}{\beta I_m}\leq v_{*}(x,y)\leq \frac{\ln x - \ln (\frac{\gamma}{\beta})}{\sigma\beta y}.\end{equation}
\end{Theorem}

\begin{proof}
Let $(S, I, V)$ be the solution of the SIVR model \eqref{a1} with $S(0)=x$ and $I(0)=y$. Letting $W=S+V$ and integrating equation \eqref{1dd} from $0$ to $t$ gives
\begin{equation} \label{ee1}
   W(0)e^{-\beta\int_{0}^{t} I(\tau){\textrm{d}\tau}}\leq  W(t) \leq W(0)e^{-\sigma\beta\int_{0}^{t} I(\tau){\textrm{d}\tau}}
            \end{equation}
for $t\geq0$.

Since $I(t)$ is continuous and increasing on $t\in [0, v_{*}(x,y)]$, $y\leq I(t)\leq I_m$ for $t\in [0, v_{*}(x,y)]$.
If $v_{*}\neq n\tau$ for $n=1,2,3,\cdots$, we have $I'(v_*)=0$ and
$(S+\sigma V)(v_{*})=\frac{\gamma}{\beta}$. In view of \eqref{ee1},
$$\frac{\gamma}{\beta}=(S+\sigma V)(v_{*}(x,y))\leq W(v_{*}(x,y))\leq xe^{-\sigma\beta\int_{0}^{v_{*}} I(\tau){\textrm{d}\tau}}\leq xe^{-\sigma\beta yv_{*}},$$
$$\frac{\gamma}{\sigma\beta}=(S/\sigma+V)(v_{*}(x,y))\geq W(v_{*}(x,y))\geq xe^{-\beta\int_{0}^{v_{*}} I(\tau){\textrm{d}\tau}}\geq xe^{-\beta I_m v_{*}}.$$
Taking the natural logarithm and rearranging gives \eqref{de1}.

If $v_{*}=n_0 \tau$ for $n_0=1,2,3,\cdots$,
we have $I'(v^{-}_{*})\geq0$ and $I'(v^{+}_{*})\leq0$, that is,
$$\frac{\gamma}{\beta}\leq(S+\sigma V)(v_{*})\leq W(v_{*}(x,y))\ \textrm{and}\ \frac{\gamma}{\beta}\frac{1}{\min(1-\varphi+\sigma\varphi,\sigma)}\geq W(v_{*}(x,y)).$$
By \eqref{ee1},
$$\frac{\gamma}{\beta}\leq W(v_{*}(x,y))\leq xe^{-\sigma\beta\int_{0}^{v_{*}(x,y)} I(\tau){\textrm{d}\tau}}\leq xe^{-\sigma\beta yv_{*}(x,y)}$$
\textrm{and}\
$$\frac{\gamma}{\beta}\frac{1}{\min(1-\varphi+\sigma\varphi,\sigma)}\geq W(v_{*}(x,y))\geq xe^{-\beta\int_{0}^{v_{*}(x,y)} I(\tau){\textrm{d}\tau}}\geq xe^{-\beta v_{*}(x,y) I_m} .$$
Now we can take the natural logrithm of above two equations to get \eqref{de1} and complete the proof.
\end{proof}

\section{The stopping time}

Usually, we say the epidemic stops if there is no longer infected individual after a special time $t_S$, that is, $I(t_S-1)>0$ and
$I(t)=0$ for $t\geq t_S$.
It is well-known that in any actual epidemic situations, the infected number is a nonnegative integer. However,
in most epidemic compartment models described by ODEs or PDEs, $I(t)$ is a continuous function of time $t$.
We have the limit $\lim_{t\to\infty}I(t)=0$ as proved above,
but $I(t)>0$ for $t>0$, it is difficult to derive the exact time $t_S$.
 As we know, there is no definition for the stopping time.

The critical time for a given threshold $\mu(>0)$ provides us a convenient route.
Denote
$$t_S:=\overline t_\mu |_{\mu=0.5}=\sup \{t>0:\ I(t)\geq 0.5\},$$
where $(S, I, V)$ is the solution of \eqref{a1}.
It is easy to see that $I(t_S-1)\geq 1$ and
$I(t)=0$ for $t\geq t_S$ if $I(t)$ is a nonnegative integer for any $t\geq 0$.

With the above definition, we have the following estimates of the stopping time for the $SIR$ model \eqref{a0}  by using the results in \cite{HIP}.

\begin{Theorem} Let $(S,I)$ be the solution of \eqref{a0} with $S(0)=x>0, I(0)=y>0.5$. Then we have\\
$(i)$ $\frac 1\gamma \ln (\frac{x+y}{\gamma/\beta+0.5})\leq t_S\leq \frac{x+y}{0.5\gamma};$\\
$(ii)$ $t_S\leq \frac{\ln(2y)}{\gamma-\beta x }$ if $ x\in [0, \gamma/\beta);$\\
$(iii)$ $\lim\limits_{x+y\to \infty} \frac{t_S}{\ln(2(x+y))/\gamma}=1.$
\end{Theorem}
\begin{proof}
Recalling that the $SIR$ model \eqref{a0} admits a single wave, we have $u_*(x,y)=u^*(x,y):=u(x,y)$.
It is shown that in \cite{HIP}
$$\frac 1\gamma \ln (\frac{x+y}{\gamma/\beta+\mu})\leq u(x,y)\leq \frac{x+y}{\mu \gamma},$$
$$ u(x,y)\leq \frac{\ln(y/\mu)}{\gamma-\beta x }\, \textrm{ if}\,  x\in [0, \gamma/\beta),$$
$$\lim\limits_{x+y\to \infty} \frac{u(x,y)}{\ln((x+y)/\mu)/\gamma}=1.$$
Replacing $\mu$ by $0.5$, we have $t_S=u(x,y)$ and complete the proof.
\end{proof}

Let us consider the example as shown in Fig. \ref{tu2}. Since $\beta=0.0005, \sigma=0.002, \gamma=0.01$, $x=S_{0}=100, y=I_{0}=20$. we have
 $100\ln \frac {240}{41}\leq t_S\leq 24000$, which is not satisfactory. However, it is our first attempt to define and estimate the stopping time, and we look forward to future progress.

{\bf Acknowledgement:}\

We are grateful to the editorial board for their review of our work.

{\bf Declaration of Competing Interest:}\ Author(s) have no conflict of interest.

{\bf Contributions of Authors:}\ All the authors Phyu Phyu Win, Zhigui Lin and Mengyun Zhang contribute equally.

{\bf Data Availability:}\ The author confirms that the data supporting the findings of this study are
available within the manuscript.

\end{document}